\documentclass[compmod]{rsl}
\gridframe{N}
\usepackage{amsmath,amssymb,latexsym,natbib}

\volume{0}
\issue{0}

\pyear{}
\pmonth{}
\doinu{}
\usepackage{xspace}
\usepackage{mathrsfs}
\usepackage{centernot}
\usepackage{graphicx}
\usepackage{psfrag}

\newtheorem{conjecture}[thm]{Conjecture}

\newcommand{\ie}{i.e.\ }
\newcommand{\Scenarios}{\textsf{Scenarios}\xspace}

\newcommand{\figref}[1]{Fig.~\ref{#1}\xspace}

\newcommand{\MATH}[1]{\ensuremath{{#1}}\xspace}
\newcommand{\MATHIT}[1]{\MATH{\mathit{#1}}}
\newcommand{\MATHSF}[1]{\MATH{\mathsf{#1}}}
\newcommand{\MATHBIN}[1]{\MATH{\mathbin{#1}}}
\newcommand{\MATHCAL}[1]{\MATH{\mathcal{#1}}}

\newcommand{\MATHFRAK}[1]{\MATH{\mathfrak{#1}}}

\newcommand{\match}{\MATHSF{match}}

\renewcommand{\iff}{\MATHBIN{\;\leftrightarrow\;}}
\newcommand{\Iff}{\MATHBIN{\;\Longleftrightarrow\;}}
\renewcommand{\implies}{\MATHBIN{\;\rightarrow\;}}
\newcommand{\Implies}{\MATHBIN{\;\Longrightarrow\;}}

\newcommand{\scenarios}{\ensuremath{\mathcal{S}}\xspace}

\newcommand{\Model}{\MATHFRAK{M}}

\newcommand{\Law}[1]{\MATH{\mathsf{AllAgree}\langle{#1}\rangle}}

\newcommand{\SPR}{\MATHSF{SPR}}
\newcommand{\SPRS}{\MATH{\mathsf{SPR}(\scenarios)}}
\newcommand{\SPRM}{\MATH{\mathsf{SPR}_{\Model}}}
\newcommand{\SPRm}{\MATH{\mathsf{SPR}_{\Model^-}}}
\newcommand{\SPRp}{\MATH{\mathsf{SPR}_{\Model^+}}}

\newcommand{\wlDefM}{\MATH{\mathsf{wlDef}_{\Model}}}
\newcommand{\SPRIOb}{\MATH{\mathsf{SPR}_{\IOb}}}
\newcommand{\SPRB}{\MATH{\mathsf{SPR}_{\B}}}
\newcommand{\SPRBIOb}{\MATH{\mathsf{SPR}_{\B,\IOb}}}

\newcommand{\Ax}[1]{\MATHSF{Ax{#1}}}
\newcommand{\Axwl}[1]{\Ax{Wl}{(\MATHIT{#1})}}

\newcommand{\Q}{\MATHIT{Q}}
\newcommand{\B}{\MATHIT{B}}
\newcommand{\W}{\MATHSF{W}}
\newcommand{\IOb}{\MATHIT{IOb}}
\newcommand{\Ph}{\MATHSF{Ph}}
\newcommand{\M}{\MATHSF{M}}
\newcommand{\speed}{\MATHSF{speed}}
\newcommand{\vel}{\MATHSF{vel}}

\newcommand{\ev}{\MATHSF{ev}}
\newcommand{\wl}{\MATHSF{wline}}
\newcommand{\w}{\MATHSF{w}}

\newcommand{\obsA}{\MATH{k}}
\newcommand{\obsB}{\MATH{h}}

\title[Three Formalisations of Einstein's Relativity Principle]
      {Three Different Formalisations of Einstein's Relativity Principle}

\author[J. X. Madar{\'a}sz, G. Sz{\'e}kely]
       {JUDIT X. MADAR{\'A}SZ, GERGELY SZ{\'E}KELY}
  \affil{Alfr\'ed R\'enyi Institute of Mathematics, Hungarian Academy of
Sciences
\\[2.5pt]
{\rm and}\\[-5pt]}
\author[M. Stannett]
{MIKE STANNETT}
\affil{Department of Computer Science, The University of Sheffield
}

\leftrunninghead{judit x madar\'asz, mike stannett \& gergely sz\'ekely}
\rightrunninghead{Three Formalisations of Einstein's Relativity Principle}
\begin{document}

\maketitle

\begin{abstract}
We present three natural but distinct formalisations of Einstein's special principle of relativity, and demonstrate the relationships between them. In particular, we prove that they are logically distinct, but that they can be made equivalent by introducing a small number of additional, intuitively acceptable axioms.
\end{abstract}

\section{Introduction}
\label{sec:introduction}

The special principle of relativity (\SPR), which states that the laws of physics should be the same in all inertial frames, has been foundational to physical thinking since the time of Galileo, and gained renewed prominence as Einstein's first postulate of relativity theory 
\citep[\S1]{Ein16}:

\begin{quote}
If a system of coordinates K is chosen so that, in relation to it, physical laws hold good in their simplest form, the \emph{same} laws hold good in relation to any other system of coordinates K' moving in uniform translation relatively to K. This postulate we call the 
``special principle of relativity.''
\end{quote}

Despite its foundational status, the special principle of relativity remains problematic due to its inherent ambiguity 
\citep{Sza04,GS13,GS13b}. What, after all, do we mean by ``physical laws'', and what does it mean to say that the ``same laws'' hold in two different frames?

These ambiguities often lead to misunderstandings and misinterpretations of
the principle. See, e.g., \cite{Muller92} for the resolution of one such misinterpretation.  We believe that formalisation is the best way to
eliminate these ambiguities.  In this paper we investigate the principle of
relativity in an axiomatic framework of mathematical logic.  However, we
will introduce not one but three different naturally arising versions of the
principle of relativity, not counting the parameters on which they depend,
such as the formal language of the framework used.

It is not so surprising, when one tries to capture \SPR formally, that more than one ``natural'' version offers itself -- not only was Einstein's description of his principle given only informally, but its roots reach back to Galileo's even less formal ``ship story'' \citep[pp.~186--187]{Galileo}.

Since all three of the versions we investigate are ``natural'', and simply reflect different approaches to capturing the original idea, there is no point trying to decide which is the ``authentic'' formalisation. The best thing we can do is to investigate how the different formalisations are related to each other. Therefore, in this paper we investigate under which assumptions these formalisations become equivalent.  

In a different framework but with similar motivations, G\" om\"ori and Szab\'o also introduce several formalisations of Einstein's ideas \citep{Sza04,GS13,GS13b,Gom15}. Intuitively, what they refer to as ``covariance'' corresponds to our principles of relativity and what they call the principle of relativity is an even stronger assumption. However, justifying this intuition is beyond the scope of this paper, as it would require us to develop a joint framework in which both approaches can faithfully be interpreted.

\subsection{Contribution}
In this paper we present three logical interpretations ($\SPRM$, $\SPR^+$ and \SPRBIOb) of the relativity principle, and investigate the extent to which they are equivalent. We find that the three formalisations are logically distinct, although they can be rendered equivalent by the introduction of additional axioms. We prove rigorously the following relationships.

\subsubsection{Counter-examples and implications requiring no additional axioms (\figref{fig:fig1})}
\begin{itemize}
\item $\SPR^+ \centernot\Longrightarrow  \SPRM$
	(Thm.~\ref{thm:not1(+)-2})
\item $\SPRBIOb  \centernot\Longrightarrow  \SPRM$
	(Thm.~\ref{thm:not1(+)-2})
\item $\SPRBIOb  \centernot\Longrightarrow   \SPR^+$
	(Thm.~\ref{thm:not1(+)-2})
\item $ \SPRM \Implies \SPR^+$ 
	(Thm.~\ref{thm:2-1.+})
\item $ \SPRM \Implies \SPRBIOb $ 
	(Thm.~\ref{thm:2-3})
\end{itemize}
\begin{figure}[!htb]
\centering
\psfrag{Aut}[c][c]{\SPRM}
\psfrag{SPRBIOb}[c][c]{\SPRBIOb}
\psfrag{SPR+}[c][c]{$\SPR^+$}
\psfrag{thm4}[rt][rt]{Thm.\ref{thm:not1(+)-2}}
\psfrag{thm3}[lb][lb]{Thm.\ref{thm:not1(+)-2}}
\psfrag{thm5}[lt][lt]{Thm.\ref{thm:not1(+)-2}}
\psfrag{thm1}[lb][lb]{Thm.\ref{thm:2-3}}
\psfrag{thm2}[rt][rt]{Thm.\ref{thm:2-1.+}}
\includegraphics[width=0.5\columnwidth]{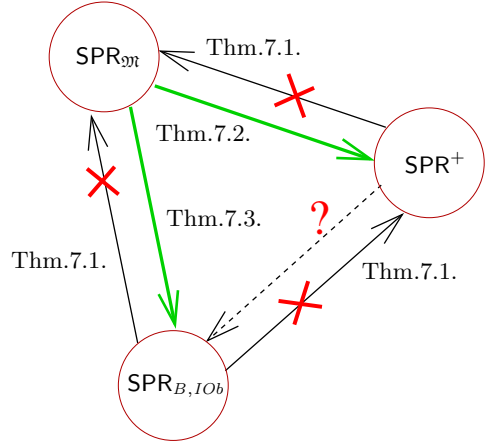}
\caption{Counter-examples and implications requiring no additional axioms.}
\label{fig:fig1}
\end{figure}

\subsubsection{Adding axioms to make the different formalisations equivalent (\figref{fig:fig2})}
\begin{itemize}
\item $\SPR^+ \Implies  \SPRBIOb$  assuming \Ax{Id}, \Ax{Ev}, \Ax{IB}, \Ax{Field} 
	(Thm. \ref{thm:1.+-3.BIOb})
\item  $\SPRBIOb \Implies \SPRM$ assuming $\mathcal{L}=\mathcal{L}_0$, \Ax{Ev}, \Ax{Ext}
	(Thm. \ref{thm:3-2})
\item $\SPR^+\Implies\SPRM$ assuming $\mathcal{L}=\mathcal{L}_0$, \Ax{Id}, \Ax{IB}, \Ax{Field}, \Ax{Ev}, \Ax{Ext}
	(Thms. \ref{thm:1.+-3.BIOb}, \ref{thm:3-2})
\item $\SPRBIOb \Implies\SPR^+$ assuming $\mathcal{L}=\mathcal{L}_0$, \Ax{Ev}, \Ax{Ext}
	(Thms. \ref{thm:2-1.+},~\ref{thm:3-2})
\end{itemize}
\begin{figure}[!htb]
\centering
\psfrag{Aut}[c][c]{\SPRM}
\psfrag{SPRBIOb}[c][c]{\SPRBIOb}
\psfrag{SPR+}[c][c]{$\SPR^+$}
\psfrag{cimke1}[l][l]{\shortstack[l]{\scriptsize $\mathcal{L}=\mathcal{L}_0$,\\ 
\scriptsize \Ax{Id}, \Ax{IB}, \Ax{Field}, \Ax{Ev},  \Ax{Ext} }}
\psfrag{cimke2}[rb][rb]{\shortstack[r]
{\scriptsize  $\mathcal{L}=\mathcal{L}_0$,\\ 
\scriptsize  \Ax{Ev}, \Ax{Ext}}}

\psfrag{Thm7}[rt][rt]{\shortstack[r]{Thm.\ref{thm:3-2}\\
\scriptsize ${}$\\
\scriptsize  $\mathcal{L}=\mathcal{L}_0$,\\ 
\scriptsize  \Ax{Ev}, \Ax{Ext}}}
\psfrag{Thm8}[lt][lt]{\shortstack[l]{Thm.\ref{thm:1.+-3.BIOb}\\ \scriptsize \Ax{Id}, \Ax{Ev}, \Ax{IB}, \Ax{Field}}}\includegraphics[width=.6\textwidth]{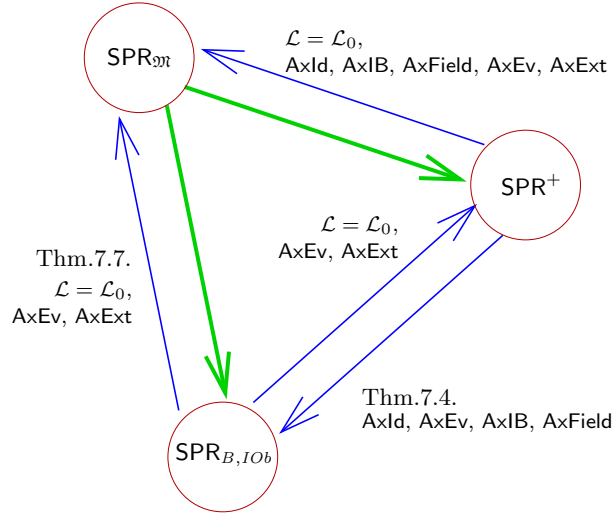}

\caption{Axioms required to make the different formalisations equivalent.}
\label{fig:fig2}
\end{figure}

\subsubsection{\SPRM, $\SPR^+$ and the decomposition of \SPRBIOb into \SPRIOb and \SPRB (\figref{fig:fig3})}
\begin{itemize}
\item  $\SPR^+ \Implies  \SPRIOb$  assuming \Ax{Id},  \Ax{Ev}
	(Thm.~\ref{thm:1.S-3.IOb})
\item  $\SPR^+ \Implies  \SPRB$  assuming \Ax{IB}, \Ax{Field}
	(Thm.~\ref{thm:1.S-3.B})
\end{itemize}
\begin{figure}[!htb]
\centering
\psfrag{Aut}[c][c]{\SPRM}
\psfrag{SPR+}[c][c]{$\SPR^+$}
\psfrag{SPRB}[c][c]{$\SPRB$}
\psfrag{SPRIOb}[c][c]{$\SPRIOb$}
\psfrag{Thm10}[rb][rb]{\shortstack[r]{Thm.\ref{thm:1.S-3.B}\\ \scriptsize\Ax{IB}, \Ax{Field}}}
\psfrag{Thm11}[lt][lt]{\shortstack[l]{Thm.\ref{thm:1.S-3.IOb}\\ \scriptsize\Ax{Id}, \Ax{Ev}}}
\includegraphics[width=0.5\textwidth]{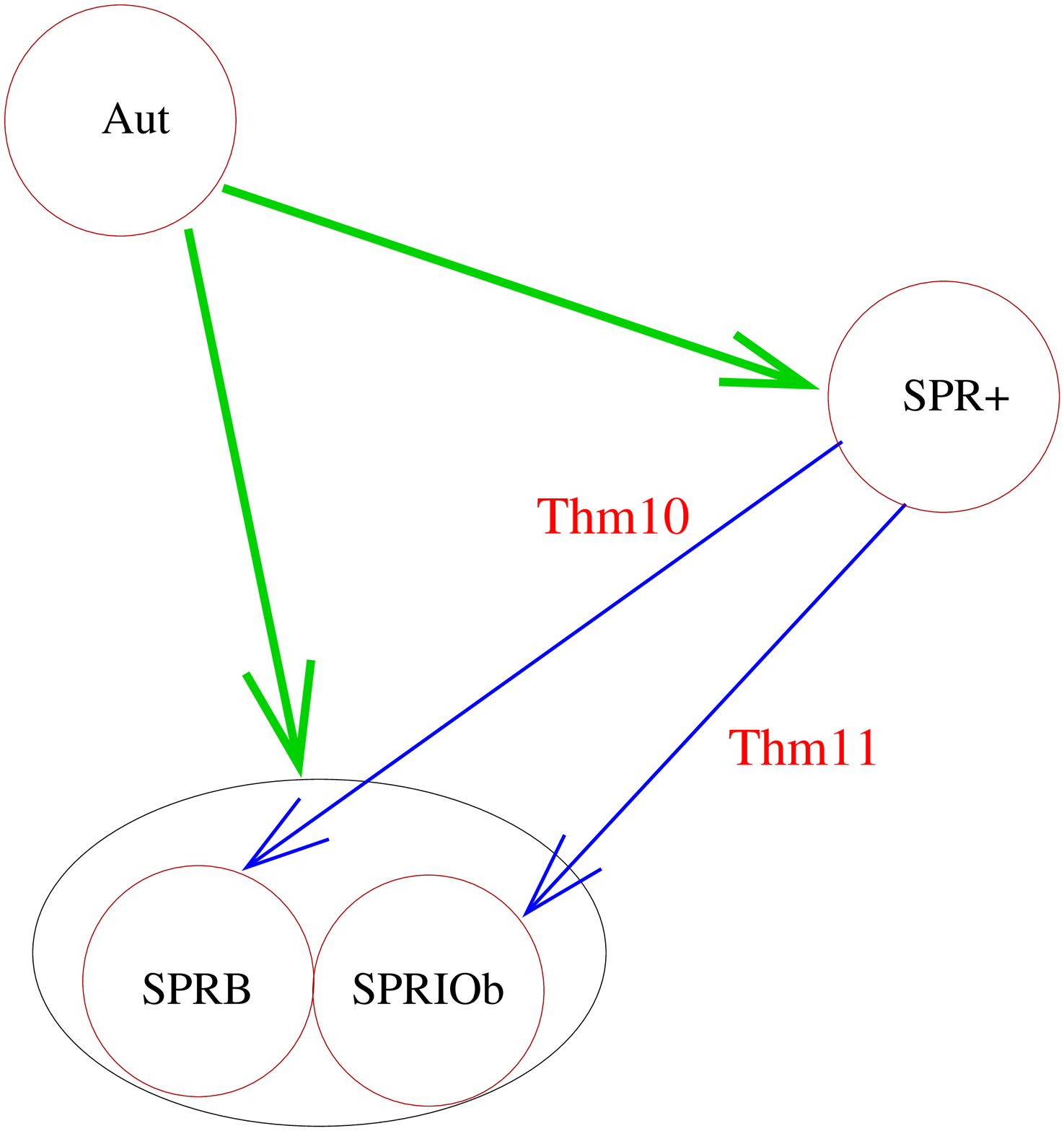}
\caption{\SPRM, $\SPR^+$ and the decomposition of \SPRBIOb into \SPRIOb and \SPRB.}
\label{fig:fig3}
\end{figure}

\paragraph{Outline of the paper.}
We begin in \ref{sec:laws-of-nature} by characterising what we mean by a
``law of nature'' in our first-order logic framework. Rather than going into
all the difficulties of defining what a law of nature \emph{is}, we focus
instead on the requirement that all inertial observers must agree as to the
outcomes of experimental scenarios described by such a law. 
In \ref{sec:examples} we give some examples, in 
\ref{sec:formalisations} we demonstrate our three formalisations of SPR, 
in \ref{sec:models} we discuss the types of  models our language admits, and 
in \ref{sec:axioms} we state the axioms that
will be relevant to our results.  These results are stated formally in \ref{sec:results}  
In \ref{sec:alternatives-to-AxIB} we discuss some
alternative assumptions to axiom \Ax{IB}. 
The proofs of the theorems can be found in \ref{sec:proofs}
We conclude with a discussion of
our findings in \ref{sec:discussion}, where we also highlight questions
requiring further investigation.

\section{Laws of Nature}
\label{sec:laws-of-nature}

Before turning our attention towards formalising the principle of
relativity, we need to present the framework in which our logical formalisms
will be expressed.  
 Following the approach described in
\citep{AMN07,AMNS11}, we will use the first-order logical (FOL) 3-sorted
language
\[
  \mathcal{L}_0 = \{ \IOb, \B, \Q, 0,1,+, \cdot, \W \}
\]
as a core language for kinematics. In this language
\begin{itemize}
\item $\IOb$ is the sort of \emph{inertial observers} (for labeling coordinate systems);
\item $\B$ is the sort of \emph{bodies}, \ie things that move;
\item $\Q$ is the sort of \emph{quantities}, \ie numbers, with 
constants $0$ and $1$,
\emph{addition} ($+$) and \emph{multiplication} ($\cdot$);
\item $\W$ is the \emph{worldview relation}, a 6-ary relation of type $\IOb \times \B \times \Q^4$. 

The statement $\W(\obsA,b,p)$ represents the idea that ``inertial observer $\obsA$ coordinatises body $b$ to be at spacetime location $p$.''
\end{itemize}

Throughout this paper we use $\obsB$, $\obsA$  and their 
variants to represent inertial observers (variables of sort $\IOb$); we use $b$ and 
$c$ to represent bodies (variables of sort $\B$); and $p$, $q$ and $r$ are 
variables of type $\Q^4$.
  The sorts of other variables will be clear from 
context.

Given this foundation, various derived notions can be defined:
\begin{itemize}
\item The event observed by $\obsA \in \IOb$ as occurring at $p \in \Q^4$ is the set of bodies that $\obsA$ coordinatises to be at $p$:
  \[ \ev_{\obsA}(p) \equiv \{b \in \B : \W(\obsA,b,p) \} . \]
  
\item For each $\obsA, \obsB \in \IOb$, the worldview 
transformation
$\w_{\obsA\obsB}$ is a binary relation on $\Q^4$ which captures the idea that $\obsB$ coordinatises at $q \in \Q^4$ the same event that $\obsA$ coordinatises at $p \in \Q^4$:
\begin{equation}
  \w_{\obsA\obsB}(p,q) \equiv [\ev_{\obsA}(p) = \ev_{\obsB}(q)]. 
  \tag{\w.def}
\end{equation}
  
\item The worldline of $b \in \B$ as observed by $\obsA \in \IOb$ is the set of locations $p \in \Q^4$ at which $\obsA$ coordinatises $b$:
  \[ \wl_{\obsA}(b) \equiv \{ p : \W(\obsA,b,p) \} . \]
\end{itemize}

We try to choose our primitive notions as simple and ``observationally
oriented'' as possible, cf.\ \citet[p.31]{Fri83}. 
Therefore the set of events is not primitive, but rather a defined concept,
\ie an event is a set of bodies that an observer observes at a certain
point of its coordinate system.
Motivation for such a definition of event goes back to Einstein and can be
found in \citet[p.6]{MTW73} and \citet[p.153]{Ein16trans}.

Since laws of nature stand or fall according to the outcomes of physical
experiments, we next consider statements, $\phi$, which describe
experimental claims.  For example, $\phi$ might say ``if this equipment has
some specified configuration today, then it will have some expected new
configuration tomorrow''.  This is very much a dynamic process-oriented
description of experimentation, but since we are using the language of
spacetime, the entire experiment can be described as a static
four-dimensional configuration of matter in time and space.  
We therefore introduce the concept of 
\emph{scenarios}, \ie sentences describing both the initial conditions and the outcomes of 
experiments.   Although our scenarios 
are primarily intended to capture experimental configurations and 
outcomes, they can also describe more complex 
situations, as illustrated by the examples in
\ref{sec:examples} 
One of our formalisations of \SPR will be the assertion that 
\emph{all} inertial
observers \emph{agree}
as to whether or not certain situations are realizable. Our definition of
scenarios is
motivated by the desire to have a suitably large set of sentences describing
these
situations.

To introduce scenarios formally, let 
us fix a language \MATHCAL{L} containing our core language
$\mathcal{L}_0$.  We will say that a formula 
\[
\phi \equiv
\phi(\obsA,\bar{x})\equiv\phi(\obsA,x_1,x_2,x_3,\dots)
\]
of language
\MATHCAL{L} describes a \emph{scenario} provided it has a single free
variable $k$ of sort \IOb (to allow us to evaluate the scenario for different
observers), and none of sort \B.  The other free variables $x_1, \dots$ can be thought of as
experimental parameters, allowing us to express such statements as
$\phi(k,v) \equiv$ ``$k$ can see some body $b$ moving with speed $v$''. 
Notice that numerical variables (in this case $v$) can sensibly be included
as free variables here, but bodies cannot -- if we allow the use of specific
individuals (\textit{Thomas}, say) we can obtain formulae (``$k$ can see
\textit{Thomas} moving with speed $v$'') which manifestly violate \SPR, since
we cannot expect \emph{all} observers $k$ to agree on such an assertion.
The truth values of certain formulas containing bodies as free
variables can happen to be independent of inertial observers,
for example $\nu_2$ in \ref{sec:examples}, but we prefer to treat these as 
exceptional cases  to be
proven from the principle of relativity and the rest of the axioms.

Thus $\phi \equiv \phi(\obsA,\bar{x})$ represents a scenario provided
\begin{itemize}
\item $\obsA$ is free in $\phi(\obsA, \bar{x})$,
\item $\obsA$ is the \emph{only} free variable of sort $\IOb$, 
\item the free variables $x_i$ are of sort $\Q$ (or any other sort of \MATHCAL{L} representing mathematical objects), and 
\item there is no free variable of sort \B (or any other sort of \MATHCAL{L} representing physical objects).
\end{itemize}
The set of all scenarios will be denoted by \Scenarios.

Finally, for any formula $\phi(\obsA,\bar{x})$ with free variables $k$
of sort $\IOb$ and $x_1,x_2,\ldots$ of any sorts, the formula
\[
  \Law\phi \quad \equiv \quad (\forall \obsA, \obsB \in \IOb)
\big((\forall \bar{x})[\phi(\obsA,\bar{x}) \iff \phi(\obsB,\bar{x})] \big)
\]
captures the idea that 
for every evaluation of the free variables $\bar x$  
  all inertial observers agree on the truth value of $\phi$.
Let us note that $\Law{}$ is defined not just for
scenarios, e.g., it is defined for the non-scenario examples of
\ref{sec:examples}, too.

In \ref{sec:formalisations}, one of the
formalisations of \SPR  will be that $\Law\phi$ holds
for every possible scenario $\phi$.

\section{Examples}
\label{sec:examples}
Here we give  examples for both scenarios and non-scenarios.  
To be able to show interesting examples beyond the core language used in
this paper, let us expand our language with a 
unary relation $\Ph$ of \emph{light signals (photons)} of type $\B$ and a 
function $\M:\B\rightarrow\Q$ for \emph{rest mass}, i.e.\ $M(b)$ is the rest mass
of body $b$. For illustrative purposes we focus in particular on 
\emph{inertial bodies}, \ie bodies moving with uniform linear motion, 
and introduce the notations 
$\speed_k(b)=v$  and  $\vel_k(b)=(v_1,v_2,v_3)$ to indicate
that $b$ is an inertial body moving with speed $v\in\Q$ and velocity
$(v_1,v_2,v_3)\in\Q^3$
 according to inertial observer $k$.
These notions can be easily defined assuming \Ax{Field} introduced in 
\ref{sec:axioms}, and their
definitions can be found, e.g., in \cite{AMNS08,MSS14}.

In the informal explanations of the examples below we freely use such 
modal
expressions  as ``can set down'' and ``can send out'', in place of ``coordinatise'', to make the
experimental idea behind scenarios intuitively clearer, and to illustrate
how the dynamical aspects of making experiments are captured in
our static framework. See also \citep{MoSze15} for a framework where the distinction
between actual and potential bodies is elaborated within first-order modal
logic.

\medskip

\textbf{Examples for scenarios:}
\begin{itemize}
\item Inertial observer $\obsA$ can 
set down a body at spacetime
location
$(0,0,0,0)$: 
\[
\phi_1(\obsA)\equiv (\exists b) \W(\obsA,b,0,0,0,0).
\]
\item Inertial observer $\obsA$ can send out an inertial body with speed $v$:
\[
\phi_2(\obsA,v)\equiv (\exists b)\speed_\obsA(b)=v.
\]
\item Inertial observer $\obsA$ can send out an inertial body at 
location $(x_1,x_2,x_3,x_4)$ with velocity $(v_1,v_2,v_3)$ and rest mass $m$:
\begin{multline*}
\phi_3(\obsA,x_1,x_2,x_3,x_4,v_1,v_2,v_3,m) \equiv\\
 (\exists b)
[W(\obsA,b,x_1,x_2,x_3,x_4)\land \vel_\obsA(b)=(v_1,v_2,v_3)\land\M(b)=m].
\end{multline*}.
\item The speed of every light signal is $v$ according to inertial observer $\obsA$: 
\[
\phi_4(\obsA,v)\equiv (\forall b)[\Ph(b)\rightarrow \speed_\obsA(b)=v].
\]
\end{itemize}

Let us consider scenario $\phi_4$ to illustrate that 
$\Law{\phi}$ means that all inertial observers agree on the truth 
value of $\phi$ for every evaluation of the free variables $\bar x$. 
Assume that the speed of light is $1$ for
every inertial observer, i.e.\ that $(\forall \obsA)(\forall
b)[\Ph(b)\rightarrow\speed_\obsA(b)=1]\land(\exists b)\Ph(b)$ holds. Then the truth
value of $\phi_4(\obsA,1)$ is true for every inertial observer $\obsA$, 
but the truth value of $\phi_4(\obsA,a)$ is false for every inertial observer 
$\obsA$ if
$a\neq 1$. Thus $\Law{\phi_4}$ holds.

\medskip
\textbf{Examples for non-scenarios:}

\begin{itemize}
\item  The speed of inertial body $b$ according to inertial observer $k$ is $v$:
\[
\nu_1(k,v,b)\equiv \speed_k(b)=v.
\]
\end{itemize}

Then $\Law{\nu_1}$ means that all inertial observers agree on the
speed of each body.
Obviously, we do not want 
such statements to hold.

Notice, incidentally, that it is possible for all observers to agree 
as to the truth value of a non-scenario, but this is generally 
something we need to prove, rather than assert a priori. 
For example, consider the non-scenario:
\begin{itemize}
\item
The speed of light signal $b$ is $v$ according to inertial observer $k$:
\[
\nu_2(k,v,b)\equiv \Ph(b)\rightarrow \speed_k(b)=v.
\]
\end{itemize}

Then $\Law{\nu_2}$ means
that all inertial observers agree
on the speed of each light signal, and it
happens to follow from
$\Law{\phi_4}$, where scenario $\phi_4$ is given above.
Therefore, \Law{\nu_2}  will follow from our formalisations of \SPR
which entail the truth of formula \Law{\phi_4}.

While different observers 
agree on the speed of any given photon in special relativity theory, 
they do not agree as to its direction of motion, which is captured by the
following non-scenario:
\begin{itemize}
\item The velocity of light signal $b$ is $(v_1,v_2,v_3)$ according to inertial
observer
$k$:
\[
\nu_3(k,v_1,v_2,v_3,b) \equiv \Ph(b)\rightarrow \vel_k(b)=(v_1,v_2,v_3).
\] 
\end{itemize}

 Then $\Law{\nu_3}$ means that all inertial observers
agree on the velocity of each light signal, 
and once again we do not want such a formula to hold.

\section{Three formalisations of SPR}
\label{sec:formalisations}

\paragraph{First formalisation.}
A natural interpretation of the special principle is to identify a set \scenarios of scenarios on which all inertial observers should agree (\ie those scenarios we consider to be experimentally relevant). If we now define
\[
  \SPRS \equiv \{ \Law\phi : \phi \in \scenarios \} ,
  \tag{1.\scenarios} 
\]
the principle of relativity becomes the statement that every formula in \SPRS holds. For example, if we assume that all inertial observers agree on \emph{all} scenarios, and define
\[
  \SPR^+ \equiv \{ \Law\phi : \phi\in\Scenarios \} ,
  \tag{1.+}
\]
then we get a ``strongest possible'' version of \SPRS formulated in the language \MATHCAL{L}.

It is important to note that the power of \SPRS (and hence that of $\SPR^+$) strongly depends on which language \MATHCAL{L} we use. It matters, for example, whether we can only use \MATHCAL{L} to express scenarios related to kinematics, or whether we can also discuss particle dynamics, electrodynamics, etc.  The more expressive \MATHCAL{L} is, the stronger the corresponding principle becomes.

\paragraph{Second formalisation.} A natural indirect approach is to assume
that the worldviews of any two inertial observers are identical.  In other
words, given any model $\Model$ of our language \MATHCAL{L}, and given any
observers $\obsA$ and $\obsB$, we can find an automorphism of the model
which maps $\obsA$ to $\obsB$, while leaving all quantities 
(and
elements of all the other sorts of \MATHCAL{L} representing mathematical 
objects)
fixed.  That is, if the only sort of \MATHCAL{L} representing 
mathematical objects is $\Q$,
we require the statement \[
  \SPRM \equiv (\forall \obsA, \obsB \in \IOb)(\exists \alpha \in
  Aut(\Model)) [\alpha(\obsA) = \obsB \land
  \alpha\mathclose{\upharpoonright_Q} = Id_Q] \tag{2}
\]
to hold, where $Id_Q$ is the identity function on $\Q$, and $\alpha\mathclose{\upharpoonright_Q}$ denotes the restriction of $\alpha$ to the quantity part of the model. 
If \MATHCAL{L} has other sorts 
representing mathematical objects than \Q, then in (2) we also require
$\alpha\mathclose{\upharpoonright_U} = Id_U$ to hold for any such sort $U$.

\paragraph{Third formalisation.}
Another way to characterise the special principle of relativity is to assume that all inertial observers agree as to how they stand in relation to bodies and each other (see \figref{fig:SPRB,SPRIOb}). In other words, we require the formulae
\begin{equation}
  \SPRB \equiv (\forall \obsA,\obsA')(\forall b)(\exists b') [ \wl_{\obsA}(b) = \wl_{\obsA'}(b') ] 
  \tag{3.\B}
\end{equation}
and
\begin{equation}
  \SPRIOb \equiv (\forall \obsA,\obsA')(\forall \obsB)(\exists \obsB')[ \w_{\obsA\obsB} = \w_{\obsA'\obsB'} ] 
  \tag{3.\IOb}
\end{equation}
to be satisfied. We will use the following notation:
\[\SPRBIOb = \{\SPRB, \SPRIOb\}.\]

\SPRBIOb is only a ``tiny kinematic slice'' of \SPR. It says
that two inertial observers are indistinguishable by possible world-lines 
and by their relation to other observers.

\begin{figure}[!htb]
\centering
\psfrag{k}[r][r]{$\forall k$}
\psfrag{h}[r][r]{$\forall h$}
\psfrag{k'}[r][r]{$\forall k'$}
\psfrag{h'}[r][r]{$\exists h'$}
\psfrag{b}[l][l]{$\forall b$}
\psfrag{b'}[l][l]{$\exists b'$}
\psfrag{w}[r][r]{$\w_{kh}$}
\psfrag{w'}[r][r]{$\w_{k'h'}$}
\includegraphics[width=0.6\textwidth]{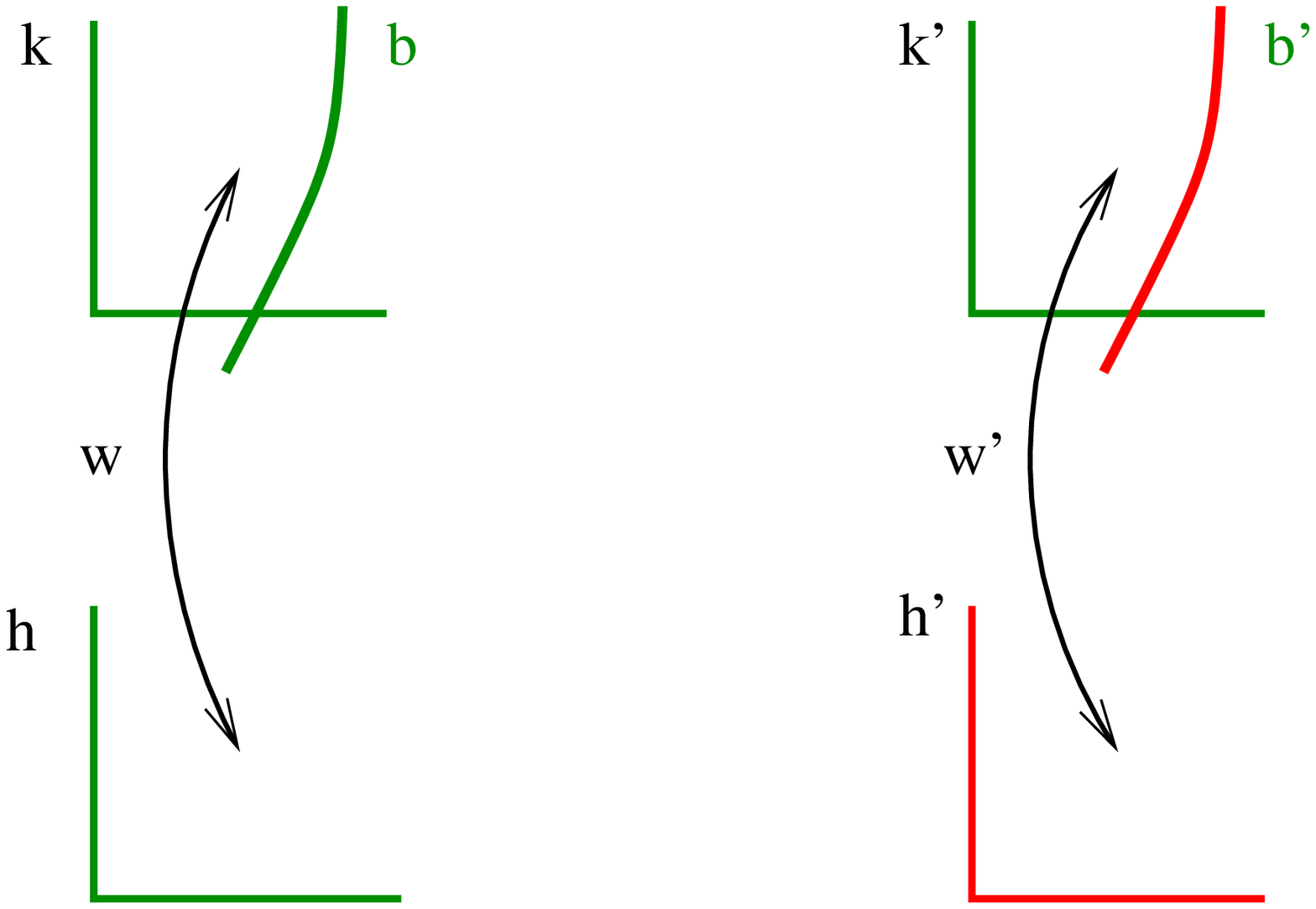}
\caption{Illustration for \SPRB and \SPRIOb.}
\label{fig:SPRB,SPRIOb}
\end{figure}

\section{Models satisfying SPR}
\label{sec:models}

According to Theorems~\ref{thm:2-1.+} and \ref{thm:2-3} below, \SPRM is the
strongest of our three formalisations of the relativity principle, since it
implies the other two without any need for further assumptions.

The `standard' model of special relativity satisfies \SPRM and
therefore also $\SPR^+$ and \SPRBIOb,  where by
the `standard' model we mean a model determined up to isomorphisms
by the following properties: (\textit{a}) the structure of quantities is
isomorphic to that of real numbers; (\textit{b}) all the worldview
transformations are Poincar\'e transformations; (c) for every inertial observer
$k$ and Poincar\'e transformation $P$, there is another observer $h$, such
that $\w_{kh}=P$; (\textit{d}) bodies can move exactly on the smooth
timelike and lightlike curves; and (\textit{e}) worldlines uniquely determine
bodies and worldviews uniquely determine inertial observers.

In fact, there are several models satisfying 
\SPRM in the literature, and these
models also ensure that the axioms used in this paper are mutually
consistent.  Indeed, in \citep{Sze13,AMNSS14,SM14} we have demonstrated
several extensions  
of the `standard' model of special relativity which
satisfy \SPRM.  Applying the methods used in those papers, it is not
difficult to show that \SPRM is also consistent with classical kinematics. 
Again, this is not surprising as there are several papers in the literature
showing that certain formalisations of the principle of relativity cannot
distinguish between classical and relativistic kinematics, and as
\cite{Ig10,Ig11} has shown, when taken together with other
assumptions, \SPR implies that the group of transformations between inertial
observers can only be the Poincar\'e group or the inhomogeneous Galilean
group. For  further developments of this theme, see
\citep{LM76,Bo,Pal03,AM15}.

The simplest way to get to special relativity from \SPRM is to extend the
language $\mathcal{L}_0$ with light signals and assume Einstein's light
postulate, \ie light signals move with the same speed in every direction
with respect to \emph{at least one} inertial observer.  Then, by \SPRM,
\Ax{Ph} follows, \ie light signals move with the same speed in every
direction according to \emph{every} inertial observer.  \Ax{Ph}, even
without any principle of relativity, implies (using only some trivial
auxiliary assumptions such as \Ax{Ev} (see p.~\pageref{axev})), that the
transformations between inertial observers are Poincar\'e transformations;
see, e.g., \citep[Thm.2.2]{AMNS11}.

It is worth noting that \SPRM also admits models which extend the `standard'
model of special relativity, for example models containing faster-than-light
bodies which can interact dynamically with one another
\citep{Sze13,AMNSS14,SM14}.

\section{Axioms}
\label{sec:axioms}

We now define various auxiliary axioms. As we show below, whether or not two formalisations of SPR are equivalent depends to some extent on which of these axioms one considers to be valid.

In these axioms, the spacetime origin is the point $e_0 = (0,0,0,0)$, and the unit points along each axis are defined by $e_1 = (1,0,0,0)$, $e_2 = (0,1,0,0)$, $e_3 = (0,0,1,0)$ and $e_4 = (0,0,0,1)$. We call $e_0, \dots, e_4$ the \emph{principal
locations}.

\begin{description}
\item[\Ax{Ev}]\label{axev} All observers agree as to what can be observed. If $\obsA$ can observe an event somewhere, then $\obsB$ must also be able to observe that event somewhere:
  \[ (\forall \obsA,\obsB)(\forall p)(\exists q)[ \ev_{\obsA}(p) = \ev_{\obsB}(q) ] . \]
\end{description}

\begin{description}  
\item[\Ax{Id}] If $\obsA$ and $\obsB$ agree as to what's happening at each
of the 5 principal 
locations, then they agree as to what's happening everywhere
(see \figref{fig:AxId}):
\[ (\forall \obsA,\obsB)
  \big(
   (\forall i \in \{0,\dots,4\})[ \ev_{\obsA}(e_i) = \ev_{\obsB}(e_i) ] 
  \implies 
   (\forall p)[ \ev_{\obsA}(p) = \ev_{\obsB}(p) ]  \big) .
  \]
We can think of this axiom as a generalised form of the assertion that all 
worldview  transformations are 
affine transformations.
\begin{figure}
\centering
\psfrag{k}[l][l]{$k$}
\psfrag{h}[r][r]{$h$}
\psfrag{e1}[rt][rt]{$e_1$}
\psfrag{e2}[lt][lt]{$e_2$}
\psfrag{e3}[t][t]{$e_3$}
\psfrag{e0}[rb][rb]{$e_0$}
\psfrag{p}[lb][lb]{$p$}
\includegraphics[width=0.7\textwidth]{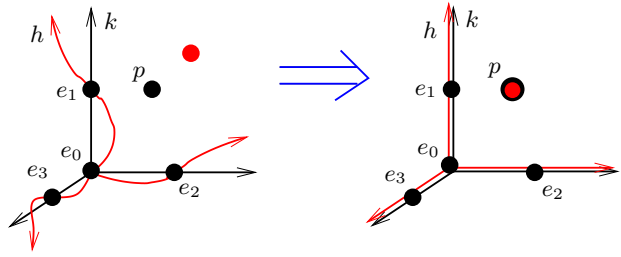}
\caption{Schematic representation of \Ax{Id}: If $\obsA$ and $\obsB$ agree as to what's happening at each of the 5 principal locations, then they agree as to what's happening at every location $p$.}
\label{fig:AxId}
\end{figure}

\item[\Ax{ExtIOb}] If two inertial observers coordinatise exactly the same events at every possible location, they are actually the same observer:
  \[ (\forall \obsA,\obsB)
    \big(  (\forall p)[ \ev_{\obsA}(p) = \ev_{\obsB}(p) ] \implies \obsA = \obsB \big) . \]

\item[\Ax{ExtB}] If two bodies have the same worldline (as observed by any observer $\obsA$), then they are actually the same body:
  \[ (\forall \obsA)(\forall b,b')
    [ \wl_{\obsA}(b) = \wl_{\obsA}(b') \implies b = b' ] . \]

\item[\Ax{Ext}] We write this as shorthand for $\Ax{ExtB} \land \Ax{ExtIOb}$.

\item[\Ax{Field}] $(\Q, 0,1,+, \cdot)$ satisfies the most fundamental properties of $\mathbb{R}$, \ie it is a field (in the sense of abstract algebra; see, e.g., \citep{Ste09}).
\label{axfield}
\end{description}

Notice that we do not assume a priori that \Q is the 
field $\mathbb{R}$ of real numbers, because we do not know and cannot determine experimentally whether the structure of quantities in the real world is isomorphic to that of $\mathbb{R}$. 
Moreover, using arbitrary fields makes our findings more general.

\begin{description}
\item[\Ax{IB}] All bodies (considered) are inertial, \ie their worldlines are straight lines according to every inertial observer:
  \[ (\forall \obsA)(\forall b)(\exists p,q) 
[ q \neq e_0 \land\wl_{\obsA}(b) = \{ p + \lambda q : \lambda \in \Q \} ] .
\]
\end{description}

\Ax{IB} is a strong assumption. In section 
\ref{sec:alternatives-to-AxIB}, we introduce
generalisations of \Ax{IB} allowing accelerated bodies, too. We choose to
include \Ax{IB} in our main theorems because it is arguably the simplest and 
clearest of these generalisations. The main generalisation \wlDefM of \Ax{IB} 
is a meta-assumption and the others are
quite technical assertions which are easier to  understand in 
relation to \Ax{IB}.

\section{Results}
\label{sec:results}

If \Model is some model for a FOL language $\mathcal{L}$, and $\Sigma$ 
is some collection of logical formulae in that language, we write 
$\Model \vDash \Sigma$ to mean that every $\sigma \in \Sigma$ is valid when 
interpreted within \Model. If $\Sigma_1, \Sigma_2$ are both collections of 
formulae, we write $\Sigma_1 \vDash \Sigma_2$ to mean that 
\[
  \Model \vDash \Sigma_2 \quad \text{ whenever } \quad \Model \vDash \Sigma_1
\]
holds for every model \Model of $\mathcal{L}$. For a general 
introduction to logical models, see
\citep{Men15,Mar02}.

Theorem~\ref{thm:not1(+)-2} demonstrates that our three formalisations of
the principle of relativity are logically distinct.  It is worth noting that
based on the ideas used in the proof of Theorem~\ref{thm:not1(+)-2} it is
also easy to construct sophisticated counterexamples to their
equivalence extending the
`standard' model of special relativity.

%

\begin{thm} The formalisations \SPRM, $\SPR^+$ and \SPRBIOb are logically distinct:
\label{thm:not1(+)-2}
\end{thm}
\begin{itemize}
\item $\Model \vDash \SPR^+ ~~ \centernot\Longrightarrow ~~ \SPRM$. $\square$
\item $\Model \vDash \SPRBIOb ~~ \centernot\Longrightarrow ~~ \SPRM$. $\square$
\item $\Model\vDash\SPRBIOb ~~ \centernot\Longrightarrow ~~
\Model\vDash\SPR^+$. $\square$
\end{itemize}

By Theorems~\ref{thm:2-1.+} and \ref{thm:2-3}, \SPRM is the strongest version of the three formalisations since it implies the other two without any extra assumptions.

\begin{thm}
$ \SPRM \Implies \Model \vDash \SPR^+$. $\square$
\label{thm:2-1.+}
\end{thm}

\begin{thm}
$ \SPRM \Implies \Model \vDash  \SPRBIOb $. $\square$
\label{thm:2-3}
\end{thm}


Theorem \ref{thm:1.+-3.BIOb} tells us that $\SPR^+$ can be made as powerful as \SPRBIOb by adding additional axioms.  This is an immediate consequence of Theorems \ref{thm:1.S-3.IOb} and \ref{thm:1.S-3.B}

\begin{thm} 
$\SPR^+ \cup \{ \Ax{Id}, \Ax{Ev}, \Ax{IB}, \Ax{Field} \} \vDash \SPRBIOb $. $\square$
\label{thm:1.+-3.BIOb}
\end{thm}

  \begin{thm} There exist scenarios $\psi, \widetilde\psi$ such that
  $
    \SPR(\psi,\widetilde\psi) \cup \{ \Ax{Id}, \Ax{Ev} \} \vDash \SPRIOb
  $. $\square$
  \label{thm:1.S-3.IOb}
  \end{thm}
  
  \begin{thm}
  There exists a scenario $\xi$ such that
  $
    \SPR(\xi) \cup \{ \Ax{IB}, \Ax{Field} \} \vDash \SPRB
  $. $\square$
  \label{thm:1.S-3.B}
  \end{thm}


Theorem \ref{thm:3-2} tells us that equipping \SPRBIOb with additional axioms allows us to recapture the power of \SPRM (and hence, by Theorem \ref{thm:2-1.+}, $\SPR^+$).

\begin{thm} Assume $\mathcal{L}=\mathcal{L}_0$. Then
$\Model \vDash \SPRBIOb \cup \{  \Ax{Ev}, \Ax{Ext} \}$ $\Implies$ $\SPRM$. $\square$
\label{thm:3-2}
\end{thm}

Thus, although \SPRM, \SPRBIOb and $\SPR^+$ are logically distinct, they become equivalent in the presence of suitable auxiliary axioms.

\section{Alternatives to \Ax{IB}}
\label{sec:alternatives-to-AxIB}

In this section, we generalise \Ax{IB} to allow discussion of 
accelerated bodies.

For every model $\Model$ we formulate 
a property which says that world-lines are parametrically
definable subsets of $\Q^4$, where the parameters can be chosen only from
\Q. For the definitions, cf.\ \cite[\S1.1.6, \S1.2.1]{Mar02}. 
\begin{description}
\item[\wlDefM] For any $\obsA\in\IOb$ and $b\in\B$ there is a formula
$\varphi(y_1,y_2,y_3,y_4,x_1,\ldots, x_n)$, where all the free variables
$y_1,y_2,y_3,y_4,x_1,\ldots,x_n$ of $\varphi$ 
are of sort $\Q$, and there is
$\bar a\in\Q^n$ such that
$\wl_\obsA(b)\equiv \{q\in\Q^4 :  
\Model\vDash\varphi(q,\bar a)\}$.
\end{description}

\medskip

We note that plenty of curves in $\Q^4$ are definable in the sense above, e.g.\ curves 
which can be defined by polynomial functions,  as well as the worldlines of uniformly accelerated bodies
 in both special relativity and Newtonian kinematics.

In general, not every accelerated worldline is definable -- indeed, the set of 
curves which are definable depends both on the language and the model. For example,
uniform circular motion is undefinable in many models; however, if
we extend the language with the sine function
as a primitive notion and assert its basic
properties by including the appropriate axioms, then uniform circular 
motion becomes definable.

By Theorems \ref{thm:1.S-3.Bii} and \ref{thm:1.+-3.BIObii}, 
assumptions \Ax{IB} and \Ax{Field} can be
replaced by \wlDefM in Theorem \ref{thm:1.+-3.BIOb} (Theorem \ref{thm:1.+-3.BIObii} follows immediately from Theorems \ref{thm:1.S-3.IOb} and \ref{thm:1.S-3.Bii})

\begin{thm}
$(\wlDefM\text{\rm\ and }\Model\vDash\SPR^+)$  $\Implies$ $\Model\vDash\SPRB$.
$\square$
 \label{thm:1.S-3.Bii}
\end{thm}

\begin{thm}
$(\wlDefM\text{\rm\ and }\Model\vDash\SPR^+ \cup \{ \Ax{Id}, \Ax{Ev}\})$ 
$\Longrightarrow$ $\Model\vDash\SPRBIOb$.
$\square$
\label{thm:1.+-3.BIObii}
\end{thm}

We note that $\Model\vDash\{\Ax{IB},\Ax{Field}\}$ $\Implies$ $\wlDefM$. 
Moreover \wlDefM is more general than \Ax{IB} assuming \Ax{Field}. The disadvantage of \wlDefM is
that it is not an axiom, but a property of model $\Model$. We now introduce,
for every natural number $n$, an axiom \Axwl{n} which is more general than 
\Ax{IB} assuming \Ax{Field} and $n\geq 3$, and stronger than \wlDefM.

\begin{description}
\item[\Axwl{n}] Worldlines are determined by $n$ distinct locations, 
\ie if two worldlines agree at $n$ distinct locations, then they
coincide:
\begin{multline*}
(\forall\obsA,\obsA')(\forall b,b')\\
{[}(\exists\text{ distinct } p^1,\ldots,p^n \in \wl_\obsA(b)\cap\wl_{\obsA'}(b'))
\ \rightarrow\ 
\wl_\obsA(b)=\wl_{\obsA'}(b')].
\end{multline*}
\end{description}

We note that $\{\Ax{IB},\Ax{Field}\}\vDash\Axwl{2}$, 
and $\Axwl{n}\vDash \Axwl{i}$ if $i\geq n$.

Furthermore, for every $n$, we have
$\Model\vDash\Axwl{n}\ \Longrightarrow\ \wlDefM$. To see that this must be 
true, assume $\Model\vDash\Axwl{n}$, choose $\obsA\in\IOb$ 
and $b\in\B$, let $p^1,\ldots, p^n\in\wl_\obsA(b)$ be distinct locations, 
and define
\begin{multline*}
\varphi(y_1,y_2,y_3,y_4,x_1^1,x_2^1,x_3^1,x_4^1,\ldots,x_1^n,x_2^n,x_3^n,x_4^n)
\equiv\\
 (\exists \obsB)(\exists c)
[(y_1,y_2,y_3,y_4),(x_1^1,x_2^1,x_3^1,x_4^1),\ldots,
(x_1^n,x_2^n,x_3^n,x_4^n)\in\wl_\obsB(c)].
\end{multline*}
Then it is easy to see that $\wl_\obsA(b)\equiv\{q\in\Q^4 : 
\Model\vDash\varphi(q,p^1,\ldots,p^n)\}$, whence $\wlDefM$ holds, as claimed.

By Theorems \ref{Axwlthm-1} and \ref{Axwlthm-2},
Theorems \ref{thm:1.+-3.BIOb} and \ref{thm:1.S-3.B} remain true if we replace
\Ax{IB} and \Ax{Field} with \Axwl{n}.

\begin{thm}
\label{Axwlthm-1}
$\SPR^+ \cup \{ \Ax{Id}, \Ax{Ev}, \Axwl{n} \} \vDash \SPRBIOb $. 
$\square$
\end{thm}

\begin{thm}
\label{Axwlthm-2}
There is  a scenario $\xi$ such that
  $\SPR(\xi) \cup \{ \Axwl{n} \} \vDash \SPRB$. $\square$
\end{thm}

\section{Proofs}
\label{sec:proofs}

We begin by proving a simple lemma which allows us to identify when two observers are in fact the same observer. This lemma will prove useful in several places below.

\bigskip
\begin{lem}
$\{ \Ax{Ev}, \Ax{ExtIOb} \} \vDash (\forall \obsA, \obsB, \obsB') [ ( \w_{\obsA\obsB} = \w_{\obsA\obsB'}) \implies \obsB = \obsB' ]$.
\label{lem:aux}
\end{lem}
\begin{proof}
Suppose $\w_{\obsA\obsB} = \w_{\obsA\obsB'}$, and choose any $p \in \Q^4$. Let $q \in \Q^4$ satisfy $\ev_{\obsA}(q) = \ev_{\obsB}(p)$, so that $\w_{\obsA\obsB}(q,p)$ holds ($q$ exists by \Ax{Ev}). Since $\w_{\obsA\obsB} = \w_{\obsA\obsB'}$, it follows that $\w_{\obsA\obsB'}(q,p)$ also holds, so that $\ev_{\obsB}(p) = \ev_{\obsA}(q) = \ev_{\obsB'}(p)$. This shows that $\obsB$ and $\obsB'$ see the same events at every $p \in \Q^4$, whence it follows by \Ax{ExtIOb} that $\obsB = \obsB'$, as claimed.
\end{proof}

\bigskip
\begin{proof}[\textbf{Proof of Theorem 
\ref{thm:not1(+)-2}}
$\SPR^+  \centernot\Longrightarrow  \SPRM$, $\SPRBIOb  \centernot\Longrightarrow  \SPRM$, $\SPRBIOb  \centernot\Longrightarrow   \SPR^+$]
~

\medskip
Constructing the required counterexamples in detail from scratch would be
too lengthy and technical, and would obscure the key ideas explaining why
such models exist.  Accordingly, here we give only the `recipes' on how to
construct these models.

To prove that $\SPR^+$ does not imply \SPRM, let  $\Model^-$ be any model satisfying
\SPRm and \Ax{Ext}, and containing at least two inertial observers and a
body $b$ such that the worldlines of $b$ are distinct according to the two
observers.  Such models exist, see, e.g., the ones constructed in
\cite{Sze13}.  The use of \Ax{Ext} ensures that distinct bodies have
distinct worldlines. 
Let us now construct an extension $\Model$ of $\Model^-$ (violating
\Ax{Ext}) by adding uncountably-infinite many copies of body $b$, as well as
countably-infinite many copies of every other body.  
Clearly, $\Model$ 
does not satisfy \SPRM, since this would require the existence of an
automorphism taking a body having uncountably many copies to one having only
countably many copies.  Nonetheless, 
$\Model$ satisfies $\SPR^+$ 
 since
it can be elementarily extended to an even larger model $\Model^+$ 
satisfying  \SPRp
(and hence, by Thm.~\ref{thm:2-1.+}, $\SPR^+$) by increasing the population
of other bodies so that every body has an equal (uncountable) number of
copies (see \cite[Theorem 2.8.20]{JuditPhD}). 
Thus  $\Model$ satisfies $\SPR^+$ but not \SPRM.
In more detail: 
Let $\Model^+$ be an extension of $\Model$ obtained
by increasing the population of bodies so that every body has an equal 
(uncountable) number of copies. We will use the Tarski-Vaught test 
\citep[Prop.2.3.5, p.45]{Mar02} to show that $\Model^+$ is an
elementary extension of $\Model$. Let $\phi(v,w_1,\ldots,w_n)$ be a formula 
and suppose $a_1,\ldots,a_n$ in $\Model$ and $d^+$ in $\Model^+$ 
satisfy 
$\Model^+\vDash\phi(d^+,a_1,\ldots,a_n)$. We have to
find a $d$ in $\Model$ such that
$\Model^+\vDash\phi(d,a_1,\ldots,a_n)$. 
If $d^+$ is not a body, then $d^+$ is already in $\Model$
since we extended $\Model$ only by bodies.
Assume, then, that $d^+$ is a body. 
Then $d^+$ has infinitely many copies in $\Model$, so we can choose
$d$, a copy of $d^+$ in $\Model$, such that
$d\not\in\{a_1,\ldots,a_n\}$.
Let $\alpha$ be any automorphism of $\Model^+$
which interchanges $d$ and $d^+$ and leaves every other element
fixed. Then $\alpha(a_1)=a_1,\ldots,\alpha(a_n)=a_n$ and $\alpha(d^+)=d$.
By $\Model^+\vDash\phi(d^+,a_1,\ldots,a_n)$, we have that
$\Model^+\vDash\phi(\alpha(d^+),\alpha(a_1),\ldots,\alpha(a_n))$. Thus
$\Model^+\vDash\phi(d,a_1,\ldots, a_n)$ as required.

To prove that \SPRBIOb does not imply \SPRM or $\SPR^+$, let $\mathfrak{M}$
be any model of $\SPRBIOb$ and $\Ax{Ext}$ containing at least two inertial
observers $k$ and $h$ and a body $b$ for which $\wl_k(b)=\{(0,0,0,t) : t\in
\Q \}\neq\wl_h(b)$.  Such models exist, see, e.g., \citep{Sze13}. 
Duplicating body $b$ leads to a model in which $\SPRBIOb$ is still satisfied
since duplicating a body does not change the possible worldlines but it
violates both \SPRM (the automorphism taking one inertial observer to
another cannot take a body having only one copy to one having two) and
$\SPR^+$ (since scenario $\varphi(m)\equiv (\forall
b,c)[\wl_m(b)=\wl_m(c)=\{(0,0,0,t) : t\in \Q \}\implies b=c]$ holds for $h$
but does not hold for $k$).  \end{proof}

\begin{proof}[\textbf{Proof of Theorem 
\ref{thm:2-1.+}}
$ \SPRM \Implies \Model \vDash \SPR^+$]
~

\medskip
Suppose \SPRM, 
so that for any observers $\obsA$ and $\obsB$ there is
an automorphism $\alpha\in Aut(\Model)$ such that $\alpha(\obsA)=\obsB$ and
$\alpha$ leaves elements of all sorts of \MATHCAL{L} representing 
mathematical objects fixed. 

We will prove that \Law\phi holds for all $\phi \in \Scenarios$, \ie given
any observers $\obsA$ and $\obsB$, any scenario $\phi$, and any set
$\bar{x}$ of parameters for $\phi$, we have \begin{equation}
  \phi(\obsA,\bar{x}) \iff \phi(\obsB,\bar{x}).  \label{eq:law}
\end{equation}

To prove this, choose some $\alpha \in Aut(\Model)$ which fixes $\Q$ 
and all the other sorts representing mathematical objects,
and satisfies $\alpha(\obsA) = \obsB$.

Suppose $\Model \vDash \phi(\obsA,\bar{x})$. Since $\alpha$ is an automorphism, $\phi(\alpha(\obsA),\alpha(\bar{x}))$ also holds in \Model. But $\alpha(\obsA) = \obsB$ and $\alpha(\bar{x}) = \bar{x}$, so this says that $\Model \vDash \phi(\obsB,\bar{x})$.  Conversely, if $\phi(\obsB,\bar{x})$ holds in \Model, then so does $\phi(\obsA,\bar{x})$, by symmetry.
\end{proof}

\bigskip

\begin{proof}[\textbf{Proof of Theorem 
\ref{thm:2-3}}
$ \SPRM \Implies \Model \vDash  \SPRBIOb $]
~

\medskip
Suppose \SPRM. Then $(\forall \obsA, \obsB)(\exists \alpha \in Aut(\Model))[\alpha(\obsA) = \obsB \land \alpha\mathclose{\upharpoonright_Q} = Id_Q]$. We wish to prove that 
\Model satisfies
\begin{eqnarray*}
  \SPRIOb &\equiv& (\forall \obsA,\obsA')(\forall \obsB)(\exists \obsB')[ \w_{\obsA\obsB} = \w_{\obsA'\obsB'}
], \\
  \SPRB &\equiv& (\forall \obsA,\obsA')(\forall b)(\exists b') [ \wl_{\obsA}(b) = \wl_{\obsA'}(b')
]. 
\end{eqnarray*}

Recall that whenever $\alpha$ is an automorphism of \Model and $R$ is a defined $n$-ary relation on \Model, we have
\begin{equation}
  R(v_1, \dots, v_n) \Iff R( \alpha(v_1), \dots, \alpha(v_n) ).
  \label{eqn:R}
\end{equation}

\underline{Proof of \SPRIOb.}
Choose any $\obsA, \obsA', \obsB$. We need to find $\obsB'$ such that $\w_{\obsA\obsB} = \w_{\obsA'\obsB'}$, so let $\alpha \in Aut(\Model)$ be some automorphism taking $\obsA$ to $\obsA'$, and define $\obsB' = \alpha(\obsB)$. Now it's enough to note that $\w_{\obsA\obsB}(p,q)$ is a defined 10-ary relation on \Model (one parameter each for $\obsA$ and $\obsB$, 4 each for $p$ and $q$), so that
\[
  \w_{\obsA\obsB}(p,q) \Iff \w_{\alpha(\obsA)\alpha(\obsB)}(\alpha(p),\alpha(q))
\]
holds for all $p, q $, by (\ref{eqn:R}). Substituting $\alpha(\obsA) = \obsA'$, $\alpha(\obsB) = \obsB'$, and noting that $\alpha$ leaves all spacetime coordinates fixed (because $\alpha\mathclose{\upharpoonright_Q} = Id_Q$) now gives
\[
  \w_{\obsA\obsB}(p,q) \Iff \w_{\obsA'\obsB'}(p,q)
\]
as required.

\underline{Proof of \SPRB.}
Choose any $\obsA, \obsA'$ and $b$. We need to find $b'$ such that $\wl_{\obsA}(b) = \wl_{\obsA'}(b')$. As before, let $\alpha \in Aut(\Model)$ be some automorphism taking $\obsA$ to $\obsA'$, define $b' = \alpha(b)$, and note that ``$p \in \wl_{\obsA}(b)$'' is a defined 6-ary relation on \Model. Applying (\ref{eqn:R}) now tells us that
\[
  \wl_{\obsA}(b) = \wl_{\alpha(\obsA)}(\alpha(b)) = \wl_{\obsA'}(b')
\]
as required.
\end{proof}

\bigskip
\begin{proof}[\textbf{Proof of Theorem 
\ref{thm:1.+-3.BIOb}}
$ \SPR^+ \cup \{ \Ax{Id}, \Ax{Ev}, \Ax{IB}, \Ax{Field} \} \vDash \SPRBIOb $]
~

\medskip
This is an immediate consequence of Theorems \ref{thm:1.S-3.IOb} and 
\ref{thm:1.S-3.B}
\end{proof}

\begin{proof}[\textbf{Proof of Theorem 
\ref{thm:1.S-3.IOb}}
$ \SPR(\psi, \widetilde\psi) \cup \{ \Ax{Id}, \Ax{Ev} \} \vDash \SPRIOb $ for some $\psi, \widetilde\psi$]
~

\medskip

Given a 5-tuple of locations $\vec{x}_i = (x_0, x_1, x_2, x_3, x_4)$, let $\match$ (see \figref{fig:match}) 
be the relation
\[
 \match(\obsB,\obsA,\vec{x}_i) \equiv
  (\forall i \in \{0,\dots,4\})(\ev_{\obsB}(e_i) = \ev_{\obsA}(x_i))
\]
and define $\psi, \widetilde\psi \in \Scenarios$ by
\begin{eqnarray*}
\psi(\obsA,p,q,\vec{x}_i) &\equiv& 
    (\exists \obsB) [ (\ev_{\obsB}(p) = \ev_{\obsA}(q))  
     \land \match(\obsB,\obsA,\vec{x}_i) ],
\\
  \widetilde\psi(\obsA,p,q,\vec{x}_i) &\equiv& 
    (\exists \obsB) [ (\ev_{\obsB}(p) \neq \ev_{\obsA}(q))  
     \land \match(\obsB,\obsA,\vec{x}_i) ].
\end{eqnarray*}
\begin{figure}
\centering
\psfrag{h}[r][r]{$h$}
\psfrag{k}[r][r]{$k$}
\psfrag{h1}[l][l]{$h$}
\psfrag{e1}[rt][rt]{$e_1$}
\psfrag{e2}[lt][lt]{$e_2$}
\psfrag{e3}[b][b]{$e_3$}
\psfrag{e0}[rb][rb]{$e_0$}
\psfrag{x0}[x0][b]{$x_0$}
\psfrag{x1}[rb][rb]{$x_1$}
\psfrag{x2}[b][b]{$x_2$}
\psfrag{x3}[lt][lt]{$x_3$}
\includegraphics[keepaspectratio, width=0.7\textwidth]{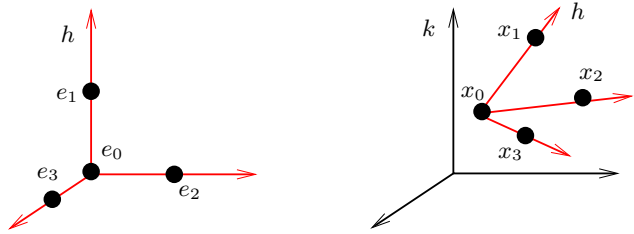}
\caption{Schematic showing the behaviour of function $\match$, which tells us which locations in \obsA's worldview correspond to the principle locations in \obsB's worldview.}
\label{fig:match}
\end{figure}

Choose any $\obsA, \obsB, \obsA'$. In order to establish \SPRIOb we need 
to demonstrate some $\obsB' $ such that $\w_{\obsA\obsB} = \w_{\obsA'\obsB'}$. To do this, let $x_i$ be such that $\ev_{\obsA}(x_i) = \ev_{\obsB}(e_i)$ for all $i = 0, \dots, 4$ (these exist by \Ax{Ev}).

Then, in particular, $\psi(\obsA,e_0,x_0,\vec{x}_i) \equiv (\exists \obsB) [ (\ev_{\obsB}(e_0) = \ev_{\obsA}(x_0)) \land \match(\obsB,\obsA,\vec{x}_i) ]$ holds. Since $\SPR(\psi,\widetilde\psi)$, it follows that $\psi(\obsA',e_0,x_0,\vec{x}_i)$ also holds, so there is some $\obsB'$ 
satisfying
\begin{equation}
  \match(\obsB',\obsA',\vec{x}_i). 
  \label{eq:obsB'}
\end{equation}

Now choose any $p$ and $q$. We will show that $\ev_{\obsB}(p) = \ev_{\obsA}(q)$ if and only if $\ev_{\obsB'}(p) = \ev_{\obsA'}(q)$, whence $\w_{\obsA\obsB} = \w_{\obsA'\obsB'}$, as claimed.

\textbf{Case 1:} 
Suppose $\ev_{\obsB}(p) = \ev_{\obsA}(q)$.
In this case, $\psi(\obsA,p,q,\vec{x}_i)$ holds, and we need to prove that $\ev_{\obsB'}(p) = \ev_{\obsA'}(q)$. It follows from $\SPR(\psi, \widetilde\psi)$ that $\psi(\obsA',p,q,\vec{x}_i)$ also holds, \ie there exists $\obsB''$ 
satisfying
\begin{equation}
(\ev_{\obsB''}(p) = \ev_{\obsA'}(q)) \land \match(\obsB'',\obsA',\vec{x}_i).
\label{eq:obsB''}
\end{equation}
It follows from (\ref{eq:obsB'}) that 
\[
\match(\obsB',\obsA',\vec{x}_i) \land \match(\obsB'',\obsA',\vec{x}_i)
\]
holds, \ie
\[
  \ev_{\obsB'}(e_i) = \ev_{\obsA'}(x_i) = \ev_{\obsB''}(e_i)
\]
holds for all $i = 0, \dots, 4$. By \Ax{Id} it follows that $(\forall r)(\ev_{\obsB'}(r) = \ev_{\obsB''}(r))$.

It now follows from (\ref{eq:obsB''}) that
\[
\ev_{\obsB'}(p) = \ev_{\obsB''}(p) = \ev_{\obsA'}(q),
\]
\ie $\ev_{\obsB'}(p) = \ev_{\obsA'}(q)$, as required.

\textbf{Case 2:} 
Suppose $\ev_{\obsB}(p) \neq \ev_{\obsA}(q)$.
In this case, $\widetilde\psi(\obsA,p,q,\vec{x}_i)$ holds, and we need to prove that $\ev_{\obsB'}(p) \neq \ev_{\obsA'}(q)$. It follows from $\SPR(\psi, \widetilde\psi)$ that $\widetilde\psi(\obsA',p,q,\vec{x}_i)$ also holds, \ie there exists $\obsB''$ 
satisfying
\begin{equation}
(\ev_{\obsB''}(p) \neq \ev_{\obsA'}(q)) \land
\match(\obsB'',\obsA',\vec{x}_i).
\label{eq:obsB''not}
\end{equation}
As before, it follows from (\ref{eq:obsB'}) that $\ev_{\obsB'}(e_i) = \ev_{\obsA'}(x_i) = \ev_{\obsB''}(e_i)$ holds for all $i = 0, \dots, 4$, and hence by \Ax{Id} that $(\forall r)(\ev_{\obsB'}(r) = \ev_{\obsB''}(r))$.

It now follows from (\ref{eq:obsB''not}) that
\[
\ev_{\obsB'}(p) = \ev_{\obsB''}(p) \neq \ev_{\obsA'}(q),
\]
\ie $\ev_{\obsB'}(p) \neq \ev_{\obsA'}(q)$, as required.
\end{proof}

\bigskip
\begin{proof}[\textbf{Proof of Theorem 
\ref{thm:1.S-3.B}}
$  \SPR(\xi) \cup \{ \Ax{IB}, \Ax{Field} \} \vDash \SPRB $ for some $\xi$]
~

\medskip
We define $\xi \in \Scenarios$ by $
  \xi(\obsA,p,q) \equiv (\exists b)(p, q \in \wl_{\obsA}(b)) .
$

To see that this satisfies the theorem, suppose that \Ax{IB} and \Ax{Field} both hold, and choose any $\obsA, \obsA' \in \IOb$ and any $b \in \B$. 
We will demonstrate a body $b'$ satisfying 
$\wl_{\obsA}(b) = \wl_{\obsA'}(b')$, whence \SPRB holds, as claimed. 

According to \Ax{IB}, there exist points $p_{\obsA},q_{\obsA}\in\Q^4$, where $q_{\obsA} \neq e_0$, and
\[
  \wl_{\obsA}(b) = \{ p_{\obsA} + \lambda q_{\obsA} : \lambda \in \Q \} .
\]
In particular, therefore, the points $p = p_{\obsA}$ and $q = p_{\obsA} + q_{\obsA}$ are distinct elements of the straight line $\wl_{\obsA}(b)$.

This choice of $p$ and $q$ ensures that the statement $\xi(\obsA,p,q)$ holds. By $\SPR(\xi)$, it follows that $\xi(\obsA',p,q)$ also holds; \ie there is some 
body $b'$ such that $p, q \in \wl_{\obsA'}(b')$. By \Ax{IB}, $\wl_{\obsA'}(b')$ is also a straight line.

It follows that $\wl_{\obsA}(b)$ and $\wl_{\obsA'}(b')$ are both straight lines containing the same two distinct points $p$ and $q$. Since there can be at most one such line (by \Ax{Field}) it follows that $\wl_{\obsA}(b) =\wl_{\obsA'}(b')$, as claimed.
\end{proof}

\bigskip
\begin{proof}[\textbf{Proof of Theorem 
\ref{thm:3-2}}
If $\MATHCAL{L} = \MATHCAL{L}_0$, then $\Model \vDash \SPRBIOb \cup \{  \Ax{Ev}, \Ax{Ext} \}$ $\Implies$ $\SPRM$]
~

\medskip
Choose any $\obsA, \obsA'$. We need to demonstrate an automorphism $\alpha \in Aut(\Model)$ which is the identity on $\Q$ and satisfies $\alpha(\obsA) = \obsA'$.

\begin{itemize}
\item \underline{Action of $\alpha$ on $\IOb$}:
Suppose $\obsB \in \IOb$. According to \SPRIOb, there exists some $\obsB'$ such that $\w_{\obsA\obsB} = \w_{\obsA'\obsB'}$. By Lemma \ref{lem:aux}, this $\obsB'$ is uniquely defined (since we would otherwise have distinct $\obsB', \obsB''$ satisfying $\w_{\obsA'\obsB'} =\w_{\obsA'\obsB''}$). Define $\alpha(\obsB) = \obsB'$.

\item \underline{Action of $\alpha$ on $\B$}:
Suppose $b \in \B$. According to \SPRB, there exists some $b'$ such that $\wl_{\obsA}(b) = \wl_{\obsA'}(b')$. By \Ax{ExtB}, this $b'$ is uniquely defined. Define $\alpha(b) = b'$.

\item \underline{Action on $\Q$}: Define $\alpha \mathclose{\upharpoonright_Q} = Id_Q$. Notice that this forces $\alpha(p) = p$ for all $p \in \Q^4$.
\end{itemize}

We already know that $\alpha$ fixes $\Q$. It remains only to prove that $\alpha \in Aut(\Model)$, and that $\alpha(\obsA) = \obsA'$.

\medskip
\underline{Proof that $\alpha(\obsA) = \obsA'$}:  Recall first that for any equivalence relation $R$, it is the case that $R = R \circ R^{-1}$, and that given any $\obsA,\obsA',\obsA''$, we have 
\begin{itemize}
\item $\w_{\obsA\obsA}$ is an equivalence relation;
\item $\w_{\obsA\obsA'}^{-1} = \w_{\obsA'\obsA}$; and
\item $\w_{\obsA\obsA'} \circ \w_{\obsA'\obsA''} = \w_{\obsA\obsA''}$ (by \Ax{Ev} and (\w.def)).
\end{itemize}

By construction, we have $\w_{\obsA\obsA} = \w_{\obsA'\alpha(\obsA)}$, whence $\w_{\obsA'\alpha(\obsA)}$ is an equivalence relation. It now follows that
\begin{eqnarray*}
  \w_{\obsA'\alpha(\obsA)} 
    &=& \w_{\obsA'\alpha(\obsA)} \circ \w_{\obsA'\alpha(\obsA)}^{-1} \\
    &=& \w_{\obsA'\alpha(\obsA)} \circ \w_{\alpha(\obsA)\obsA'} \\
    &=& \w_{\obsA'\obsA'} .
\end{eqnarray*}
and hence, by Lemma \ref{lem:aux}, that $\alpha(\obsA) = \obsA'$, as required.

\medskip
\underline{Proof that $\alpha \in Aut(\Model)$}:
We know that $\wl_{\obsA}(b) = \wl_{\obsA'}(\alpha(b)) = \wl_{\alpha(\obsA)}(\alpha(b))$, or in 
other words, given any $b$ and $q \in \Q^4$,
\begin{equation}
  b \in \ev_{\obsA}(q) \Iff \alpha(b) \in \ev_{\alpha(\obsA)}(q) .
\label{eq:*}
\end{equation}
We wish to prove $\W(\obsB,b,p) \iff \W(\alpha(\obsB),\alpha(b),p)$ for all $\obsB$, $b$ and $p$ (recall that $\alpha(p) = p$). This is equivalent to proving
\begin{equation}
b \in \ev_{\obsB}(p) \Iff \alpha(b) \in \ev_{\alpha(\obsB)}(p) .
\label{eq:Delta}
\end{equation}
Choose any $h,b,p$.
Let $q \in \Q^4$ satisfy $\ev_{\obsA}(q) = \ev_{\obsB}(p)$ -- such a $q$ exists by \Ax{Ev}. Then 
\begin{equation}
  \ev_{\alpha(\obsA)}(q) = \ev_{\alpha(\obsB)}(p)
\label{eq:Cat}
\end{equation}
because $\w_{\obsA\obsB} = \w_{\obsA'\alpha(\obsB)}= \w_{\alpha(\obsA)\alpha(\obsB)}$ by the definition of $\alpha(\obsB)$ and the fact, established above, that $\obsA' = \alpha(\obsA)$.

To prove (\ref{eq:Delta}), let us first assume that $b \in \ev_{\obsB}(p)$. Then $b \in \ev_{\obsA}(q)$, so $\alpha(b) \in \ev_{\alpha(\obsA)}(q)$ by (\ref{eq:*}). Then, by (\ref{eq:Cat}) we have $\alpha(b) \in \ev_{\alpha(\obsB)}(p)$. Next, choose $b \not\in \ev_{\obsB}(p) = \ev_{\obsA}(q)$. Then $\alpha(b) \not\in \ev_{\alpha(\obsA)}(q) = \ev_{\alpha(\obsB)}(p)$ by (\ref{eq:*}) and (\ref{eq:Cat}). It follows that $b \in \ev_{\obsB}(p) \Iff \alpha(b) \in \ev_{\alpha(\obsA)}(q)$, as required.

\bigskip
It remains to show that $\alpha$ is a bijection.

\medskip
\underline{Proof of injection}:
\begin{itemize}
\item \textit{Observers}: Suppose $\alpha(\obsB) = \alpha(\obsB')$. 
Then, by the definition of $\alpha$, we have
\[
  \w_{\obsA\obsB} = \w_{\obsA'\alpha(\obsB)} = \w_{\obsA'\alpha(\obsB')} =
\w_{\obsA\obsB'}
\]
and now $\obsB = \obsB'$ by Lemma \ref{lem:aux}
\item \textit{Bodies}: Suppose $\alpha(b) = \alpha(c)$. By definition of $\alpha$ we have
\[
  \wl_{\obsA}(b) = \wl_{\obsA'}(\alpha(b)) = \wl_{\obsA'}(\alpha(c)) = \wl_{\obsA}(c)
\]
and now $b = c$ follows by \Ax{ExtB}.
\end{itemize}

\underline{Proof of surjection}: We need for every $\obsB', b'$ that there are $\obsB, b$ satisfying $\alpha(\obsB) = \obsB'$ and $\alpha(b) = b'$.

\begin{itemize}
\item \textit{Observers}: Let $\obsB' \in \IOb$. By \SPRIOb there exists $\obsB$ such that $\w_{\obsA'\obsB'} = \w_{\obsA\obsB}$, and now $\obsB' = \alpha(\obsB)$ for any such $\obsB$.
\item \textit{Bodies}: Let $b' \in \B$. By \SPRB there exists $b \in \B$ such that $\wl_{\obsA'}(b') = \wl_{\obsA}(b)$, and now $b' = \alpha(b)$ for any such $b$.
\end{itemize}

This completes the proof.
\end{proof}

\begin{proof}[\textbf{Proof of Theorem \ref{thm:1.S-3.Bii}}
$(\wlDefM\text{\rm\ and }\Model\vDash\SPR^+)$  $\Implies$ 
	$\Model\vDash\SPRB$]
~

\medskip
Assume  $\Model\vDash\SPR^+$ and \wlDefM.
Choose any $\obsA, \obsA' \in \IOb$ and any $b \in \B$.
We will demonstrate a body $b'$ satisfying  $\wl_{\obsA}(b) = 
\wl_{\obsA'}(b')$.

Let $\varphi(\bar y,\bar x)\equiv \varphi(y_1,y_2,y_3,y_4,x_1,\ldots, x_n)$ 
be a formula such that all the free variables $\bar y,\bar x$ of $\varphi$ 
are of sort \Q, and choose
$\bar a\in\Q^n$ such that
\[
\wl_k(b)\equiv \{q\in\Q^4 :
\Model\vDash\varphi(q,\bar a)\},
\]
Such $\varphi$ and $\bar a$ exist by \wlDefM.

We define $\psi\in\Scenarios$ by
$\psi(\obsB,\bar x)\equiv
(\exists c)(\forall q)[q\in\wl_{\obsB}(c)\leftrightarrow
\varphi(q,\bar x)]$.
Clearly, $\Model\vDash\psi(\obsA,\bar a)$. Then, by $\SPR^+$,
$\Model\vDash\psi(\obsA',\bar a)$. Thus, there is $b'\in B$ such
that
$\wl_{\obsA'}(b')\equiv\{q\in\Q^4 : \Model\vDash \varphi(q,\bar a)\}$, and 
 $\wl_\obsA(b)=\wl_{\obsA'}(b')$ for this $b'$.
\end{proof}

\begin{proof}[\textbf{Proof of Theorem \ref{thm:1.+-3.BIObii}} 
$(\wlDefM\text{\rm\ and }\Model\vDash\SPR^+ \cup \{ \Ax{Id},
\Ax{Ev}\})$ $\Implies$\\  $\Model\vDash\SPRBIOb$]
~

\medskip
This is an immediate consequence of 
Theorems \ref{thm:1.S-3.IOb} and \ref{thm:1.S-3.Bii}
\end{proof}

\begin{proof}[\textbf{Proof of Theorem 
\ref{Axwlthm-1}}
$ \SPR^+ \cup \{ \Ax{Id}, \Ax{Ev}, \Axwl{n} \} \vDash \SPRBIOb $]
~

\medskip
This is an immediate consequence of Theorems \ref{thm:1.S-3.IOb} and 
\ref{Axwlthm-2}
\end{proof}

\begin{proof}[\textbf{Proof of Theorem \ref{Axwlthm-2}} $\SPR(\xi) \cup 
\{ \Axwl{n} \} \vDash \SPRB $ for some $\xi$]
~
\medskip

We define $\xi \in \Scenarios$ 
by $\xi(\obsA,p^1,\ldots,p^n) 
\equiv (\exists b)(p^1,\ldots,p^n \in \wl_{\obsA}(b))$. It is easy to check 
that $\xi$ satisfies the theorem, and we omit the details.
\end{proof}

\section{Discussion and Conclusions}
\label{sec:discussion}

In this paper we have shown formally that adopting different viewpoints can lead to different, but equally `natural', formalisations of the special principle of relativity. The idea that different formalisations exist is, of course, not new, but the advantage of our approach is that we can investigate the formal relationships between different formalisations, and deduce the conditions under which equivalence can be restored.

We have shown, in particular, that the model-based interpretation of the principle, \SPRM, is strictly stronger than the alternatives $\SPR^+$ and \SPRBIOb, and have identified various counterexamples to show that the three approaches are not, in general equivalent.  
On the other hand, equivalence is restored in the presence of various
axioms. We note, however, that the following question remains open, since it is unclear whether $\SPR^+$ is enough, in its  own right, to entail \SPRBIOb.

\begin{conjecture}
$\SPR^+ \centernot\Implies \SPRB$.
\end{conjecture}

An interesting direction for future research would be to investigate the extent to which our existing results can be strengthened by removing auxiliary axioms. For example, our proof that \SPRM can be recovered from $\SPR^+$ currently relies on  $\MATHCAL{L} = \MATHCAL{L}_0$, \Ax{Id}, \Ax{IB}, \Ax{Field}, \Ax{Ev} and \Ax{Ext}. While we know that \emph{some} additional axiom(s) must be required (since we have presented a counterexample showing that $\SPR^+ \centernot\Implies \SPRM$), the question remains whether we can develop a proof that works over \emph{any} language, \MATHCAL{L}, or whether the constraint  $\MATHCAL{L} = \MATHCAL{L}_0$ is required.   Again, assuming we allow the same auxiliary axioms, how far can we minimise the set \scenarios of scenarios while still entailing the equivalence between \SPRS and \SPRM?

\subsection*{Acknowledgement} The authors would like to thank the anonymous
referees for their insightful comments.

\bibliographystyle{rsl}
\bibliography{mssSPR}
\vspace*{10pt}
\address{ALFR\'ED R\'ENYI INSTITUTE OF MATHEMATICS\\ 
HUNGARIAN ACADEMY OF
SCIENCES \\
\hspace*{9pt}P.O. BOX 127 \\
\hspace*{18pt}BUDAPEST 1364, HUNGARY\\
{\it E-mail}: madarasz.judit@renyi.mta.hu, szekely.gergely@renyi.mta.hu\\
{\it URL}: http://www.renyi.hu/$\sim$madarasz, http://www.renyi.hu/$\sim$turms\\ [8pt]
DEPARTMENT OF COMPUTER SCIENCE \\ 
\hspace*{9pt}THE UNIVERSITY OF SHEFFIELD \\
\hspace*{18pt}211 PORTOBELLO, SHEFFIELD S1 4DP, UK \\
{\it E-mail}: m.stannett@sheffield.ac.uk}\\
{\it URL}: http://www.dcs.shef.ac.uk/$\sim$mps
\clearpage
\end{document}